\documentclass[12pt,reqno]{amsart}
\allowdisplaybreaks[4]

\usepackage [latin1]{inputenc}
\usepackage{graphicx}
\usepackage{amssymb}
\usepackage{color}
\usepackage{epstopdf}
\usepackage{amsthm,amsmath,amssymb}
\usepackage{mathrsfs}
\usepackage{cite}
\usepackage{subeqnarray}
\usepackage{cases}

\newtheorem{proposition}{Proposition}[section]
\newtheorem{corollary}{Corollary}[section]

\newtheorem{thm}{Theorem}[section]

\begin{document}

\begin{center}
{\large \sc \bf On the discrete modified KP hierarchy: tau functions, Fay identity and squared eigenfunction symmetries }

\vskip 20pt

{\large Kelei Tian, Guangmiao Lai, Ge Yi* and Ying Xu}

\vskip 20pt

{\it
School of Mathematics, Hefei University of Technology, Hefei 230601, China
 }

\bigskip

$^*$ Corresponding author: ge.yi@hfut.edu.cn
\bigskip

{\today}

\end{center}

\bigskip
\bigskip
\textbf{Abstract:} In this paper, we prove the existence of tau functions of the discrete modified KP hierarchy and define the squared eigenfunction symmetry. Meanwhile, the Fay identity with its difference form, the squared eigenfunction potentials and the symmetry flow acting on tau functions are obtained.

\bigskip

\bigskip
\bigskip

\section{\sc \bf Introduction}

The study of discrete systems and their integrability has made great progress in the last decades, and has promoted the development of new mathematical tools such as discrete complex analysis and discrete differential geometry, which provides an effective way for the current study of difference equations and the general theory of discrete systems\cite{CaseA,AblowitzA,SakaiR,NijhoffL}. More recently, with the gradual formation of the method of discrete integrable systems, the related research on difference operator theory and complex analysis has been greatly developed\cite{ZhangH,KonR,ZhangD}. It has successively boosted the development of the extended discrete integrable systems, solid-state physics and crystal chemistry\cite{MogilnerH,TakasakiO,YaoA}.

Hirota pioneered the discretization of bilinear equations and obtained a series of discrete integrable systems, thus facilitating the discretization process of Sato's theory\cite{HirotaN,DateM}. The KP(Kadomtsev-Petviashvili) hierarchy and modified KP hierarchy as one of the key topics of integrable system\cite{DickeyS,DateN,KashiT,JimboS}, we also concern about the discretization of these hierarchies. Based on the difference operator $\Delta$ and the shift operator $\varGamma$, the tau function of the discrete hierarchy is described by making an appropriate shift in tau function, namely replacing $\tau(t_1,t_2,t_3,...)$ by
\begin{align*}
\tau_{n}(n;t_1,t_2,t_3,...)=\tau(t_1+n,t_2-\frac{n}{2},t_3+\frac{n}{3},...),n\in \mathbb{Z}.
\end{align*}
This sequence of tau functions \{$\tau_{n}$\} corresponds to the Segal-Wilson Grassmannian\cite{HaineC}. The Hamiltonian structures, the squared eigenfunction symmetries, the additional symmetries and the gauge transformation operators for the discrete KP hierarchy have been studied\cite{KuperD,LiuT,LiuS}. Furthermore, the extended discrete KP hierarchy, the algebraic structure of the discrete KP hierarchy  and the Virasoro type algebraic structure of the constrained discrete KP hierarchy have also been obtained\cite{YaoA,SunS,LiM2013}.

Compared to the well-established theory of the discrete KP hierarchy, the dmKP(discrete modified KP) case has not quite developed so far. Under the structure of Sato theory, in this paper we give the description of the dmKP hierarchy and its integrable properties. With the definitions of the Lax equation, dressing operator and wave function, we investigate the existence of tau functions, the Fay identity and its difference form, the spectral representation and the squared eigenfunction symmetry.

The organization of the paper is as follows. In section 2, we give a brief description of the dmKP hierarchy and prove the existence of tau functions. In section 3, by introducing the basic properties and vertex operators and proving the Fay identity of the dmKP hierarchy with its difference form, we derive the spectral representation involving the product of the eigenfunctions and the adjoint eigenfunctions. In section 4, the squared eigenfunction symmetry flow acting on the two tau functions are obtained. At last, some conclusions and discussions are presented in section 5.

\bigskip

\section{\sc \bf  Tau functions of the dmKP hierarchy}

Consider the algebra $G$ of the pseudo-difference operators\cite{HaineC,LiuS,LiG}
\begin{align*}
G=\{ \sum_{i\ll \infty} u_i(n)\Delta^{i} \},
\end{align*}
where $u_{i}(n)=u_{i}(n,t_1,t_2,t_3,...)$; $n\in \mathbb{Z},t_{i}\in \mathbb{R}$.
The shift operator acting on the function $g(n)$ is defined by
\begin{align*}
\varGamma g(n)=g(n+1).
\end{align*}
The difference operator $\Delta$ acting on the function is defined as
\begin{align*}
\Delta g(n)=(\varGamma -I)g(n)=g(n+1)-g(n).
\end{align*}
The algebraic multiplication of $\Delta^{k}$ with the multiplication operator $g$ is given by the Leibnitz
rule
\begin{align*}
\Delta^{k}\circ g=\sum_{i=0}^\infty \binom ki (\Delta^{i}g)(n+k-i)\Delta^{k-i},
\end{align*}
where $\binom ki$ is the ordinary combinatorics number. The action on the function $g(n)$ by the adjoint difference operator $\Delta^{*}$ is defined by
\begin{align*}
\Delta^{*} g(n)=(\varGamma^{-1} -I)g(n)=g(n)-g(n-1),
\end{align*}
where
\begin{align*}
\varGamma^{-1} g(n)=g(n-1).
\end{align*}
For the operators $\Delta$ and $\varGamma$, they satisfy
\begin{align*}
&\Delta\circ\varGamma=\varGamma\circ\Delta,\\
&\Delta^{*}=-\Delta\circ\varGamma^{-1},\\
&(\Delta^{-1})^{*}=(\Delta^{*})^{-1}=-\varGamma\circ\Delta^{-1}.
\end{align*}

The dmKP hierarchy in Kupershmidt-Kiso version is defined as the following equation
 \begin{align}\label{Laxeq}
\frac{\partial L}{\partial t_{i}}=[(L^{i})_{\geq1},L], n=1,2,3,\cdots
\end{align}
with the Lax operator $L\in G$ given below
\begin{align}
L(n)=\Delta+u_0(n)+u_1(n)\Delta^{-1}+u_2(n)\Delta^{-2}+\cdots.
\end{align}
In this paper, the symbols $( A)_{\geq1}$ and $( A)_{\leq0}$  denote $\sum\nolimits_{i=1}^{m} a_{i}\Delta^{i}$ and $\sum\nolimits_{i=-\infty}^{0} a_{i}\Delta^{i}$ respectively for arbitrary pseudo-differential operator
$A=\sum\nolimits_{i=-\infty}^{m} a_{i}\Delta^{i}$. Similar to the case of the discrete KP hierarchy, the Lax operator $L$ for the dmKP hierarchy can be
expressed in terms of the dressing operator $Z$,
\begin{align}
L=Z\Delta Z^{-1},\nonumber
\end{align}
where $Z$ is given by
\begin{align}
Z(n)=z_0+z_1\Delta^{-1}+z_2\Delta^{-2}+\cdots.
\end{align}
Unlike other hierarchies, here $u_0$ does not default to 0 and $z_0$ does not default to 1 in the dmKP case. The dressing operator $Z$ satisfies Sato equation
\begin{align}
\frac{\partial Z}{\partial t_{i}}=-\left(L^{i}\right)_{\leq0}Z.
\end{align}
The wave function $w(n,t,\lambda)$ and the adjoint wave function $w^{*}(n,t,\lambda)$ of the dmKP
hierarchy is defined in the following way
\begin{align}
w(n,t,\lambda)&=Z(n)(1+\lambda)^{n}e^{\xi(t,\lambda)},\label{w}\\
w^{*}(n,t,\lambda)&=(Z^{-1}(n-1)\Delta^{-1})^*(1+\lambda)^{-n}(e^{-\xi(t,\lambda)}),\label{wstar}
\end{align}
with
\begin{align}
\xi(t,\lambda)&=x\lambda+t_2\lambda^2+t_3\lambda^3+\cdots.
\end{align}
The eigenfunction $\phi$ and the adjoint eigenfunction $\psi$ are defined in the identities below
\begin{align}
\phi_{t_{i}}=&(L^{i})_{\geq1}(\phi),\\
\psi_{t_{i}}=&-(\Delta^{-1}(L^{i})^{*}_{\geq1}\Delta)(\psi).
\end{align}
Then $w(n,t,\lambda)$ and $w^{*}(n,t,\lambda)$ satisfy the bilinear identity
\begin{align}
\operatorname{res}_{\lambda}w(n,t^{'},\lambda)w^{*}(n,t,\lambda)=1,\label{bilinear}
\end{align}
which is equivalent to the modified KP hierarchy. Here $res_{\lambda}\sum\nolimits_{i}a_{i}\lambda^{i}=a_{-1}$.

The modified KP hierarchy can be viewed as the particular case of the coupled modified KP hierarchy. The
existence of two tau functions for the dmKP hierarchy is showed in the theorem below.

\begin{thm}
For the dmKP hierarchy, there exist two tau functions $\tau_0 $ and $\tau_1 $ such that the wave
and adjoint wave functions can be written as
\begin{align}
w(n,t,\lambda)&=\frac{\tau_0(n,t-[\lambda^{-1}])}{\tau_1(n,t)}(1+\lambda)^{n}e^{\xi(t,\lambda)},\label{tau0}\\
w^{*}(n,t,\lambda)&=\frac{\tau_1(n,t+[\lambda^{-1}])}{\tau_0(n,t)}\lambda^{-1}(1+\lambda)^{-n}e^{-\xi(t,\lambda)},\label{tau1}
\end{align}
where $[\lambda]=(\lambda,\lambda^2/2,\lambda^3/3,\cdots)$.
\end{thm}
\begin{proof}
Under the Miura transformation $T=z^{-1}_0$ in \cite{ShawM}, the Lax equation \eqref{Laxeq} becomes into
\begin{align}
\partial_{t_{i}}(z^{-1}_0Lz_0)=[(z^{-1}_0L^{i}z_0)_{\geq1},z^{-1}_0L^{i}z_0].
\end{align}
Therefore $z^{-1}_0Lz_0$ can be seen as the Lax operator of the discrete KP hierarchy, which means that $z^{-1}_0Z$ is one dressing operator of the discrete KP hierarchy. Thus $z^{-1}_0w(n,t,\lambda)$ can be seen as one wave function of the discrete KP hierarchy. Then there exists a tau function $\tau_0(n,t)$ such that
\begin{align}
z^{-1}_0w(n,t,\lambda)=\frac{\tau_0(n,t-[\lambda^{-1}])}{\tau_0(n,t)}(1+\lambda)^{n}e^{\xi(t,\lambda)}.
\end{align}
Defining the function $\tau_1(n,t)$ by
\begin{align}
\tau_1(n,t)=\frac{\tau_0(n,t)}{z_0(n,t)},
\end{align}
we can derive \eqref{tau0}.

Furthermore, by using another form of the wave function and adjoint wave function, we have
\begin{align}
w(n,t,\lambda)&=\hat w(n,t,\lambda)(1+\lambda)^{n}e^{\xi(t,\lambda)},\\
w^{*}(n,t,\lambda)&=\hat w^{*}(n,t,\lambda)\lambda^{-1}(1+\lambda)^{-n}e^{-\xi(t,\lambda)},
\end{align}
with
\begin{align}
\hat w(n,t,\lambda)&=z_0+z_1\lambda^{-1}+z_2\lambda^{-2}+\cdots,\label{what}\\
\hat w^{*}(n,t,\lambda)&=z_0^{*}+z_1^{*}\lambda^{-1}+z_2^{*}\lambda^{-2}+\cdots.
\end{align}
Then replacing $t^{'}_{i}$ by $t_{i}+z^{-i}/i$ and $t^{'}_{i}$ by $t_{i}$ in the bilinear identity \eqref{bilinear}, it can be obtained that
\begin{align}
1=&\operatorname{res}_{\lambda}\hat w(n,t+[z^{-1}],\lambda)\hat w^{*}(n,t,\lambda)\frac{z^{-1}}{1-\lambda z^{-1}}\notag\\
=&\hat w(n,t+[z^{-1}],z)\hat w^{*}(n,t,z).\label{bilinaerhat}
\end{align}
In the limit $z^{-1}\to 0$ we have
\begin{align*}
z_0(n,t)z^{*}_0(n,t)=1.
\end{align*}
Combining \eqref{tau0} and \eqref{bilinaerhat}, then \eqref{tau1} can be proved.
\end{proof}

\bigskip

\section{\sc \bf Fay identity and squared eigenfunction potential }
{
\setlength{\parindent}{0cm}
In this section, properties of the dmKP hierarchy with Fay identity, difference Fay identity and spectral representation are introduced, which help us to deduce the expression of squared eigenfunction potential. Firstly, with the bilinear identity relation between tau functions and wave function as well as adjoint wave function, the following theorem about Fay identity is derived.
}
\begin{thm}
(Fay identity) The tau functions of the dmKP hierarchy satisfy the following Fay identity
\begin{eqnarray}
\begin{aligned}
&s_1(s_0-s_1)(s_2-s_3)\tau_0(n,t+[s_2]+[s_3])\tau_1(n,t+[s_0]+[s_1])+cyclic(s_1,s_2,s_3) \\
=&(s_1-s_2)(s_2-s_3)(s_3-s_1)\tau_0(n,t+[s_0])\tau_1(n,t+[s_1]+[s_2]+[s_3]),
\end{aligned}
\end{eqnarray}
where $cyclic(s_1,s_2,s_3)$ is the cyclic permutation over $s_1$, $s_2$ and $s_3$.
\end{thm}
\begin{proof}
From \eqref{tau0} and \eqref{tau1}, the bilinear equation \eqref{bilinear} can be rewritten as
\begin{align}
\operatorname{res}_{\lambda}\left(\lambda^{-1}\tau_0(n,t-[\lambda^{-1}])\tau_1(n,t^{'}+[\lambda^{-1}])e^{\xi(t-t^{'},\lambda)}\right)=\tau_0(n,t^{'})\tau_1(n,t).
\end{align}
Replace $t$ as $t-y$ and $t^{'}$ as $t+y$, then
\begin{align}
res_{\lambda}\left(\lambda^{-1}\tau_0(n,t-y-[\lambda^{-1}])\tau_1(n,t+y+[\lambda^{-1}])e^{-2\xi(y,\lambda)}\right)=\tau_0(n,t+y)\tau_1(n,t-y).\label{taubi}
\end{align}
By letting $y\to \frac{1}{2}([s_0]-[s_1]-[s_2]-[s_3])$ and $t\to t+\frac{1}{2}([s_0]+[s_1]+[s_2]+[s_3])$, we can rewrite \eqref{taubi} as
\begin{align}
&\operatorname{res}_{\lambda}\left(\frac{1-\lambda s_0}{\lambda(1-\lambda s_1)(1-\lambda s_2)(1-\lambda s_3)}\tau_0(n,t+[s_1]+[s_2]+[s_3]-[\lambda^{-1}])\tau_1(n,t+[s_0]+[\lambda^{-1}])\right)\notag\\
=&\tau_0(n,t+[s_0])\tau_1(n,t+[s_1]+[s_2]+[s_3]).\label{frac}
\end{align}
Note that
\begin{align*}
&\frac{1-\lambda s_0}{\lambda(1-\lambda s_1)(1-\lambda s_2)(1-\lambda s_3)}\\
=&\frac{\lambda s_0-1}{\lambda^3(s_1-s_2)(s_2-s_3)(s_3-s_1)}\left(\frac{s_2-s_3}{1-\lambda s_1}+\frac{s_3-s_1}{1-\lambda s_2}+\frac{s_1-s_2}{1-\lambda s_3}\right),
\end{align*}
then \eqref{frac} becomes into
\begin{align*}
&\operatorname{res}_{\lambda}\frac{\lambda
s_0-1}{\lambda^3}\left(\frac{s_2-s_3}{1-\lambda s_1}+\frac{s_3-s_1}{1-\lambda s_2}+\frac{s_1-s_2}{1-\lambda s_3}\right)\\
&\times\tau_0(n,t+[s_1]+[s_2]+[s_3]-[\lambda^{-1}])\tau_1(n,t+[s_0]+[\lambda^{-1}])\\
=&(s_1-s_2)(s_2-s_3)(s_3-s_1)\tau_0(n,t+[s_0])\tau_1(n,t+[s_1]+[s_2]+[s_3]).
\end{align*}
With the identity \cite{ChengO} as
\begin{align}
\operatorname{res}_{z}\left(\sum_{i=-\infty}^\infty a_{i}(\zeta)z^{-i}\frac{1}{1-z/\zeta}\right)=\zeta\left(\sum_{i=-\infty}^\infty a_{i}(\zeta)z^{-i}\right)\mid_{z=\zeta},
\end{align}
the Fay identity of the dmKP hierarchy can be derived.
\end{proof}
Further, set $s_0=s_2=0$ and shift the time variables by $[s_1]+[s_3]$, we can get
\begin{align*}
&\tau_0(n,t+[s_1])\tau_1(n,t-[s_3])+s_3 \tau_0(n,t-[s_3])\tau_1(n,t+[s_1])\\
 =&(s_3 +1)\tau_0(n,t+[s_1]-[s_3])+\tau_1(n,t).
\end{align*}
Divide by $\tau_1(n,t)\tau_1(n+1,t)$, and with $\tau(n,t-[-1])=\tau(n+1,t)$, the difference Fay identity
can be got.

\begin{proposition}
(Difference Fay identity) The following identity holds
\begin{align}
\Delta
\frac{\tau_0(n,t-[s_3])}{\tau_1(n,t)}=s^{-1}_3\left(\frac{\tau_0(n+1,t)\tau_1(n,t-[s_3])}{\tau_1(n,t)\tau_1(n+1,t)}-\frac{\tau_0(n+1,t-[s_3])}{\tau_1(n+1,t)}\right).\label{difffay}
\end{align}
Furthermore by letting $t\to t+[s_3]$,  the identity can also be rewritten as
\begin{align}
\Delta\frac{\tau_1(n,t+[s_3])}{\tau_0(n,t)}=s^{-1}_3\left(\frac{\tau_1(n,t+[s_3])}{\tau_0(n,t)}-\frac{\tau_0(n+1,t+[s_3])\tau_1(n,t)}{\tau_0(n,t)\tau_0(n+1,t)}\right).
\end{align}
\end{proposition}
Set $s_3=\lambda^{-1}$, according to \eqref{tau0} and \eqref{difffay},
\begin{align}
\Delta
w(n,t,\lambda)&=\frac{\tau_0(n+1,t-[\lambda^{-1}])}{\tau_1(n+1,t)}(1+\lambda)^{n+1}e^{\xi(t,\lambda)}-\frac{\tau_0(n,t-[\lambda^{-1}])}{\tau_1(n,t)}(1+\lambda)^{n}e^{\xi(t,\lambda)}\notag\\
&=(1+\lambda)^{n}e^{\xi(t,\lambda)}\left(\Delta\frac{\tau_0(n,t-[\lambda^{-1}])}{\tau_1(n,t)}+\lambda\frac{\tau_0(n+1,t-[\lambda^{-1}])}{\tau_1(n+1,t)}\right)\notag\\
&=\lambda\frac{\tau_0(n+1,t)\tau_1(n,t-[\lambda^{-1}])}{\tau_1(n,t)\tau_1(n+1,t)}(1+\lambda)^{n+1}e^{\xi(t,\lambda)}.\label{deltaw}
\end{align}
In the same way,
\begin{align}
\Delta
w^{*}(n,t,\lambda)=-\frac{\tau_0(n+1,t+[\lambda^{-1}])\tau_1(n,t)}{\tau_0(n,t)\tau_0(n+1,t)}(1+\lambda)^{-n-1}e^{-\xi(t,\lambda)}.\label{deltawstar}
\end{align}
Referring to \cite{DateT,DickeyS} for the properties of tau functions of the KP hierarchy, $\tau_0$ and $\tau_1$ can be viewed as tau functions of the discrete KP hierarchy, that is,
\begin{align*}
\operatorname{res}_{\lambda}\left(\tau_0(n,t-[\lambda^{-1}])\tau_0(n,t^{'}+[\lambda^{-1}])e^{\xi(t-t^{'},\lambda)}\right)=0,\\
\operatorname{res}_{\lambda}\left(\tau_1(n,t-[\lambda^{-1}])\tau_1(n,t^{'}+[\lambda^{-1}])e^{\xi(t-t^{'},\lambda)}\right)=0.
\end{align*}
Thus $\tau_{i}(i=0,1)$ satisfies the following proposition, which is the difference Fay identity of the discrete KP hierarchy.

\begin{proposition}\label{fay}
$\tau_{i}(i=0,1)$ satisfies the difference Fay identity of the discrete KP hierarchy
\begin{align*}
&(1+s^{-1}_3)\Delta\left(\frac{\tau_{i}(n,t+[s_1]-[s_3])}{\tau_{i}(n,t)}\right)\\
=&(s^{-1}_3-s^{-1}_1)\left(\frac{\tau_{i}(n,t-[s_3])\tau_{i}(n+1,t+[s_1])}{\tau_{i}(n,t)\tau_{i}(n+1,t)}-\frac{\tau_{i}(n,t+[s_1]-[s_3])}{\tau_{i}(n,t)}\right).
\end{align*}
\end{proposition}

By setting $s_1=\lambda^{-1}$ and $s_3=\mu^{-1}$, we have
\begin{eqnarray}
\begin{aligned}
&\frac{1}{\lambda-\mu}\Delta\frac{(1+\mu)^{n}}{(1+\lambda)^{n}}e^{\xi(t,\mu)-\xi(t,\lambda)}\frac{\tau_{i}(n,t+[\lambda^{-1}]-[\mu^{-1}])}{\tau_{i}(n,t)}\\
=&\frac{(1+\mu)^{n}}{(1+\lambda)^{n+1}}e^{\xi(t,\mu)-\xi(t,\lambda)}\frac{\tau_{i}(n,t-[\mu^{-1}])\tau_{i}(n+1,t+[\lambda^{-1}])}{\tau_{i}(n,t)\tau_{i}(n+1,t)
}.
\end{aligned}
\end{eqnarray}

Having discussed the Fay identity and difference Fay identity, we now turn to give the definition of the squared eigenfunction potential for the dmKP hierarchy. For the eigenfunction $\phi$ and adjoint eigenfunction $\psi$ of the dmKP hierarchy, there exists a function $S(\phi,\Delta\psi)$, $s.t.$
\begin{align*}
S(\phi,\Delta\psi)_{\Delta}&=\phi(\Delta\psi),\\
S(\phi,\Delta\psi)_{t_{n}}&=res_{\Delta}(\varGamma\Delta^{-1}(\Delta\psi)(L^{i})_{\geq1}\phi\Delta^{-1}).
\end{align*}
And another SE potential $\hat S(\Delta\phi,\psi)$ is defined by
\begin{align*}
\hat S(\Delta\phi,\varGamma\psi)=\phi\psi-S(\phi,\Delta\psi),
\end{align*}
which satisfies the compatible equations
\begin{align*}
\hat S(\Delta\phi,\varGamma\psi)_{\Delta}&=(\Delta\phi)\varGamma\psi,\\
\hat
S(\Delta\phi,\varGamma\psi)_{t_{n}}&=res_{\Delta}(\Delta^{-1}\varGamma\psi\Delta(L^{i})_{\geq1}\Delta^{-1}(\Delta\phi)\Delta^{-1}).
\end{align*}

\begin{proposition}\label{spectral}
For the eigenfunction $\phi$ and the adjoint eigenfunction $\psi$ of the dmKP hierarchy,
\begin{align}
\phi(n,t)&=\operatorname{res}_{\lambda}\left(w(n,t,\lambda)S(\phi(n,t^{'},\Delta w^{*}(n,t^{'},\lambda)))\right),\label{phi}\\
\psi(n,t)&=\operatorname{res}_{\lambda}\left(w^{*}(n,t,\lambda)\hat S(\Delta w(n,t^{'},\lambda),\psi(n+1,t^{'}))\right).\label{psi}
\end{align}
\end{proposition}
\begin{proof}
We only prove \eqref{phi} here since the proof of \eqref{psi} is similar. Denote $A_{\alpha}=\sum\nolimits_{i\geq1} a_{i,\alpha}\Delta^{i}$, by the virtue of $\partial_{t_{i}}w(n,t,\lambda)=(L^{i})_{\geq1}(w(n,t,\lambda))$, we have
\begin{align*}
&\operatorname{res}_{\lambda}\partial^{\alpha}w(n,t,\lambda)\Delta^{-1}\phi(n,t)(\Delta
w^{*}(n,t,\lambda))\frac{(l-l^{'})^{\alpha}}{\alpha}\\
=&-\sum_{\alpha
\geq0}\operatorname{res}_{\lambda}A_{\alpha}Z(n)(1+\lambda)^{n}e^{\xi(t,\lambda)}\Delta^{-1}\varGamma\phi(n-1,t)(Z^{-1^{*}}(n-1))(1+\lambda)^{-n}e^{-\xi(t,\lambda)}\frac{(l-l^{'})^{\alpha}}{\alpha}\\
=&\sum_{\alpha\geq0}\operatorname{res}_{\Delta}A_{\alpha}Z(n)Z^{-1}(n)\phi(n,t)\Delta^{-1}\frac{(l-l^{'})^{\alpha}}{\alpha}\\
=&\phi(n,t^{'}).
\end{align*}
In the second step we have used the formula\cite{DickeyS} as
\begin{align*}
\operatorname{res}_{\lambda}(P(n)(1+\lambda)^{n}e^{\xi(t,\lambda)})(Q(n-1)(1+\lambda)^{-n}e^{-\xi(t,\lambda)})=\operatorname{res}_{\Delta}P(n)Q^{*}(n),
\end{align*}
where $P(n)$ and $Q(n)$ are pseudo-difference operators, $Q^{*}$ is the adjoint of $Q$.
\end{proof}

Next, to get the expressions of the squared eigenfunction potentials for the dmKP hierarchy, we introduce the vertex operator of the dmKP hierarchy. It is defined as
\begin{align*}
X(n,\lambda,\mu)=&\frac{(1+\mu)^{n}}{\lambda(1+\lambda)^{n}}e^{\xi(t+[\lambda^{-1}],\mu)-\xi(t,\lambda)}e^{\sum\limits_{l=1}^ {\infty}\frac{1}{l}(\lambda^{-1}-\mu^{-1})\frac{\partial}{\partial t_{l}}}\\
=&\frac{(1+\mu)^{n}}{\mu(1+\lambda)^{n}}e^{\xi(t,\mu)-\xi(t-[\mu^{-1}],\lambda)}e^{\sum\limits_{l=1}^
 {\infty}\frac{1}{l}(\lambda^{-1}-\mu^{-1})\frac{\partial}{\partial t_{l}}}+\delta(\lambda,\mu),
\end{align*}
where
\begin{align*}
\theta(\lambda)=&-\sum\limits_{l=1}^{\infty}\lambda^{l}t_{l}+\sum_{l=1\infty}\frac{1}{l}
\lambda^{-l}\frac{\partial}{\partial t_{l}},\\
\delta(\lambda,\mu)=&\frac{1}{\mu}+\sum\limits_{n=-\infty}^{\infty}(\frac{\mu}{\lambda})^{n}=\frac{1}{\lambda}\frac{1}{1-\frac{\mu}{\lambda}}+\frac{1}{\mu}\frac{1}{1-\frac{\lambda}{\mu}}.
\end{align*}
Thus we have
\begin{align}
\frac{X(n,\lambda,\mu)\tau_0(n,t)}{\tau_0(n,t)}=&\frac{(1+\mu)^{n}}{\lambda(1+\lambda)^{n}}e^{\xi(t+[\lambda^{-1}],\mu)-\xi(t,\lambda)}\frac{\tau_0(n,t+[\lambda^{-1}]-[\mu^{-1}])}{\tau_0(n,t)}\notag\\
=&(1-\frac{\mu}{\lambda})w(n,t+[\lambda^{-1}],\mu)w^{*}(n,t,\lambda),\label{x1}
\end{align}
\begin{align}
\frac{X(n,\lambda,\mu)\tau_1(n,t)}{\tau_1(n,t)}=&-\frac{(1+\mu)^{n}}{\mu(1+\lambda)^{n}}e^{\xi(t,\mu)-\xi(t-[\mu^{-1}],\lambda)}\frac{\tau_1(n,t+[\lambda^{-1}]-[\mu^{-1}])}{\tau_1(n,t)}\notag\\
=&\frac{\lambda}{\mu}w(n,t,\mu)w^{*}(n,t-[\mu^{-1}],\lambda)+\delta(\lambda,\mu).\label{x2}
\end{align}
On the other hand, when $\mid\mu\mid<\mid\lambda\mid$, by using \eqref{deltaw}-\eqref{deltawstar} and Proposition \ref{fay},
\begin{align}
\Delta\left(\frac{X(n,\lambda,\mu)\tau_0(n,t)}{\tau_0(n,t)}\right)=&\frac{\lambda-\mu}{\lambda}\frac{(1+\mu)^{n}}{(1+\lambda)^{n+1}}e^{\xi(t,\mu)-\xi(t,\lambda)}\frac{\tau_0(n,t-[\mu^{-1}])\tau_0(n+1,t+[\lambda^{-1}])}{\tau_0(n,t)\tau_0(n+1,t)}\notag\\
=&(\frac{\mu}{\lambda}-1)w(n,t,\mu)\left(\Delta w^{*}(n,t,\lambda)\right),\label{x3}
\end{align}
\begin{align}
\Delta\left(\frac{X(n,\lambda,\mu)\tau_1(n,t)}{\tau_1(n,t)}\right)=&\frac{\lambda-\mu}{\lambda}\frac{(1+\mu)^{n}}{(1+\lambda)^{n+1}}e^{\xi(t,\mu)-\xi(t,\lambda)}\frac{\tau_1(n,t-[\mu^{-1}])\tau_1(n+1,t+[\lambda^{-1}])}{\tau_1(n,t)\tau_1(n+1,t)}\notag\\
=&\frac{\lambda}{\mu}\left(\Delta w(n,t,\mu)\right) w^{*}(n+1,t,\lambda).\label{x4}
\end{align}
Combing \eqref{x1}-\eqref{x4}, as the second main theorem in this section, the expressions of the squared eigenfunction potentials are derived.
\begin{thm}
Up to a constant, the expressions of the basic squared eigenfunction potentials are listed below.
\begin{align}
&S\left(w(n,t,\mu), \Delta w^{*}(n,t,\lambda)\right)= -w(n,t+[\lambda^{-1}],\mu)w^{*}(n,t,\lambda),\\
&\hat S\left(\Delta w(n,t,\mu), w^{*}(n+1,t,\lambda)\right)= w(n,t,\mu)w^{*}(n,t-[\mu^{-1}],\lambda).
\end{align}
\end{thm}

Considering the Propsition\ref {spectral} and the above expressions of the squared eigenfunction potentials, this will naturally lead to the following corollary.
\begin{corollary}
For any eigenfunction $\phi$ and adjoint eigenfunction $\psi$ of dmKP hierarchy,
\begin{align*}
&S\left(\phi(n,t), \Delta w^{*}(n,t,\lambda)\right)= -\phi(n,t+[\lambda^{-1}])w^{*}(n,t,\lambda),\\
&\hat S\left(\Delta w(n,t,\mu),\psi(n+1,t)\right)= w(n,t,\mu)\psi(n,t-[\mu^{-1}]),\\
&S\left(w(n,t,\mu), \psi(n,t)\right)= -w(n,t,\mu)\left(\psi(n,t)-\psi(n,t-[\mu^{-1}])\right),\\
&\hat S\left(\Delta \phi(n,t), w^{*}(n+1,t,\lambda)\right)=
w(n,t,\mu)\left(\phi(n,t)-\phi(n,t+[\lambda^{-1}])\right),\\
&S\left(\phi(n,t,),\Delta \psi(n,t)\right)= \operatorname{res}_{\mu}\left(\rho(\mu)\cdot \operatorname{res}_{\lambda}(S\left(w(n,t,\mu),
\Delta w^{*}(n,t,\lambda)\right)h(\lambda))\right) ,\\
&\hat S\left(\Delta \phi(n,t,), \psi(n+1,t)\right)= - \operatorname{res}_{\mu}\left(\rho(\mu)\cdot \operatorname{res}_{\lambda}(\hat
S\left(\Delta w(n,t,\mu), w^{*}(n+1,t,\lambda)\right)h(\lambda))\right),
\end{align*}
where
\begin{align*}
\rho(\mu)&=S(\phi(n,t^{'}), \Delta w^{*}(n,t^{'},\lambda)),\\
h(\lambda)&=\hat S( \Delta w(n,t,\mu), \psi(n+1,t)).
\end{align*}
\end{corollary}

\bigskip
\section{\sc \bf Squared eigenfunction symmetry}

Let $\phi_1$, $\cdots$, $\phi_{m}$ and $\psi_1$, $\cdots$, $\psi_{m}$ be eigenfunctions and adjoint
eigenfunctions of the dmKP hierarchy. The squared eigenfunction symmetry flow for the dmKP hierarchy is defined as
\begin{align*}
\partial_{\alpha}L&=[\sum\limits_{i=1}^{m}\phi_{i}\Delta^{-1}\varGamma\psi_{i}\Delta,L].
\end{align*}
Equvilently, the symmetry on the wave operator can be got as
\begin{align*}
\partial_{\alpha}Z&=\sum\limits_{i=1}^{m}\phi_{i}\Delta^{-1}\varGamma\psi_{i}\Delta Z.
\end{align*}
Then the squared eigenfunction symmetry flow acting on the wave function $w(n,t,\lambda)$ and the adjoint wave function
$w^{*}(n,t,\lambda)$ are in the following proposition.
\begin{proposition}
The wave function $w(n,t,\lambda)$ and the adjoint wave function $w^{*}(n,t,\lambda)$ satisfy the following
equations
\begin{align}
\partial_{\alpha}w(n,t,\lambda)&=\sum\limits_{i=1}^{m}\phi_{i}\hat S\left(\varGamma\psi_{i},\Delta
w(n,t,\lambda)\right),\label{partialw}\\
\partial_{\alpha}w^{*}(n,t,\lambda)&=-\sum\limits_{i=1}^{m}\psi_{i}S\left(\phi_{i},\Delta
w^{*}(n,t,\lambda)\right).
\end{align}
\end{proposition}
\begin{proof}
The squared eigenfunction symmetry on the wave operator $Z^{-1}$ and $(Z^{-1})^{*}$ can be got as following
\begin{align*}
\partial_{\alpha}Z^{-1}&= -Z^{-1}\cdot\sum\limits_{i=1}^{m}\phi_{i}\Delta^{-1}\varGamma\psi_{i}\Delta,\\
\partial_{\alpha}(Z^{-1})^{*}&=
\Delta\varGamma^{-1}\cdot\sum\limits_{i=1}^{m}\psi_{i}\Delta^{-1}\phi_{i}(Z^{-1})^{*}.
\end{align*}
With \eqref{w} and \eqref{wstar}, taking derivative by $\partial_{\alpha}$, we have
\begin{align*}
\partial_{\alpha}w(n,t,\lambda)&= (\partial_{\alpha}Z(n))(1+\lambda)^{n}e^{\xi(t,\lambda)}\\
&=\sum\limits_{i=1}^{m}\phi_{i}\Delta^{-1}\varGamma\psi_{i}\Delta Z(1+\lambda)^{n}e^{\xi(t,\lambda)}\\
&=\sum\limits_{i=1}^{m}\phi_{i}\hat S\left(\varGamma\psi_{i},\Delta w(n,t,\lambda)\right),\\
\partial_{\alpha}w^{*}(n,t,\lambda)&=\partial_{\alpha}(Z^{-1}(n-1)\Delta^{-1})^{*}(1+\lambda)^{-n}e^{-\xi(t,\lambda)}\\
&=-\varGamma\Delta^{-1}\partial_{\alpha}(Z^{-1*}(n-1)(1+\lambda)^{-n}e^{-\xi(t,\lambda)}\\
&=\varGamma\Delta^{-1}\Delta\varGamma^{-1}\cdot\sum\limits_{i=1}^{m}\psi_{i}\Delta^{-1}\phi_{i}Z^{-1*}(n-1)(1+\lambda)^{-n}e^{-\xi(t,\lambda)}\\
&=-\sum\limits_{i=1}^{m}\psi_{i}S\left(\phi_{i},\Delta w^{*}(n,t,\lambda)\right).
\end{align*}
\end{proof}

\begin{proposition}
The eigenfunction $\phi(n,t)$ and the adjoint eigenfunction $\psi(n,t)$ satisfy the following equations
\begin{align}
\partial_{\alpha}\phi(n,t)&=\sum\limits_{i=1}^{m}\phi_{i}\hat S\left(\Delta\phi,\varGamma\psi_{i}\right),\\
\partial_{\alpha}\psi(n,t)&=-\sum\limits_{i=1}^{m}\psi_{i}S\left(\phi_{i},\Delta\psi\right).
\end{align}
\end{proposition}
\begin{proof}
From \eqref{phi} and \eqref{partialw},
\begin{align*}
\partial_{\alpha}\phi(n,t)&= \operatorname{res}_{\lambda}\partial_{\alpha}w(n,t,\lambda)S\left(\phi(n,t^{'}), \Delta
w^{*}(n,t^{'},\lambda)\right)\\
&=\sum\limits_{i=1}^{m}\phi_{i}\Delta^{-1}\varGamma\psi_{i}\Delta\operatorname{res}_{\lambda}
Z(1+\lambda)^{n}e^{\xi(t,\lambda)}S\left(\phi(n,t^{'}), \Delta w^{*}(n,t^{'},\lambda)\right)\\
&=\sum\limits_{i=1}^{m}\phi_{i}\hat S\left(\Delta\phi,\varGamma\psi_{i}\right),\\
\end{align*}
In a similar way,
\begin{align*}
\partial_{\alpha}\psi(n,t)&=\partial_{\alpha}w^{*}(n,t,\lambda)\hat S \left(\Delta w(n,t^{'},\lambda),
\varGamma\psi(n,t^{'})\right)\\
&=-\sum\limits_{i=1}^{m}\psi_{i}\Delta^{-1}\phi_{i}\Delta
\operatorname{res}_{\lambda}(Z^{-1}(n-1)\Delta^{-1})^{*}(1+\lambda)^{-n}e^{-\xi(t,\lambda)}\hat S \left(\Delta
w(n,t^{'},\lambda), \varGamma\psi(n,t^{'})\right)\\
&=-\sum\limits_{i=1}^{m}\psi_{i}S\left(\phi_{i},\Delta\psi\right).
\end{align*}
\end{proof}
Then it is natural to derive the squared eigenfunction symmetry flow acting on the tau functions of the dmKP hierarchy.
\begin{thm}
The squared eigenfunction symmetry flow of the dmKP hierarchy on its tau functions is
\begin{align}
\partial_{\alpha}\tau_0(n,t)&=\sum\limits_{i=1}^{m}S\left(\phi_{i},\Delta\psi_{i}\right)\tau_0(n,t),\\
\partial_{\alpha}\tau_1(n,t)&=-\sum\limits_{i=1}^{m}\hat
S\left(\Delta\phi_{i},\varGamma\psi_{i}\right)\tau_1(n,t).
\end{align}
\end{thm}
\begin{proof}
Comparing the coefficients of $\Delta^0$ and $\Delta^1$ in the formal series of $\partial_{\alpha}Z$, it can be obtained that
\begin{align}
\partial_{\alpha}z_0&= \sum\limits_{i=1}^{m}\phi_{i}\psi_{i}z_0,\label{delta0}\\
\partial_{\alpha}z_1&= \sum\limits_{i=1}^{m}(\phi_{i}\psi_{i}z_1-\phi_{i}(\Delta\psi_{i})z_0).\label{delta1}
\end{align}
From \eqref{tau0} and \eqref{what} we have
\begin{align}
z_0&=\frac{\tau_0(n,t)}{\tau_1(n,t)},\label{z0} \\
\frac{z_1}{z_0}&=-\partial_{x}\ln\tau_0(n,t).
\end{align}
Thus we have
\begin{align*}
\partial_{\alpha}\partial_{x}\ln\tau_0(n,t)&= -\partial_{\alpha}\frac{z_1}{z_0}\\
&=\frac{1}{z^2_0}(-\partial_{\alpha}z_1\cdot z_0+z_1\cdot{\alpha}z_0)\\
&=\sum\limits_{i=1}^{m}\phi_{i}(\Delta\psi_{i}).
\end{align*}
On the other hand,
\begin{align*}
S(\phi,\Delta\psi)_{x}&= \operatorname{res}_{\Delta}(\Delta^{-1}(\varGamma\Delta\psi)\Delta\phi\Delta^{-1})\\
&=\operatorname{res}_{\Delta}(\psi\Delta^{-1}(\varGamma\psi\Delta^{-1}-\psi\Delta^{-1})\Delta\phi\Delta^{-1}\\
&=\sum\limits_{i=1}^{m}\phi_{i}(\Delta\psi_{i}),
\end{align*}
so we can get
\begin{align*}
\partial_{\alpha}\tau_0(n,t)&=\sum\limits_{i=1}^{m}S\left(\phi_{i},\Delta\psi_{i}\right)\tau_0(n,t).
\end{align*}
As for the action on $\tau_1(n,t)$, according to \eqref{delta0}-\eqref{z0},
\begin{align*}
\partial_{\alpha}\tau_1(n,t)&=\partial_{\alpha}(\frac{\tau_0(n,t)}{z_0})=\frac{1}{z^2_0}(\partial_{\alpha}\tau_0\cdot
z_0-\tau_0\cdot\partial_{\alpha}z_0)\\
&=\frac{1}{z^2_0}\sum\limits_{i=1}^{m}\left(S(\phi_{i},\Delta\psi_{i})\tau_0z_0-\tau_0\phi_{i}\psi_{i}z_0\right)\\
&=\sum\limits_{i=1}^{m}\left(S(\phi_{i},\Delta\psi_{i})-\phi_{i}\psi_{i}\right)\tau_1\\
&=-\sum\limits_{i=1}^{m}\hat S\left(\Delta\phi_{i},\varGamma\psi_{i}\right)\tau_1(n,t).
\end{align*}
\end{proof}

\bigskip
\section{\sc \bf Conclusions and discussions}
In this paper, we give the description of the dmKP hierarchy and prove the existence of tau functions in Theorem 2.1. Next, with the relationship between tau functions and wave functions, the Fay identity and its difference form are derived in Theorem 3.1 and Propsition 3.1 seperately.  After that expressions of squared eigenfunction potentials are derived in Theorem 3.2 by obtaining the spectral representation of eigenfunction. At last, we define the squared eigenfunction symmetry of the dmKP and give the flow action on tau functions in Theorem 4.1. In particular, we would like to point out that compared with the KP hierarchy, the existence of two tau functions and the discrete variable brings much difference in the dmKP case. When the discrete variable $n\to 0$, the dmKP hierarchy reduces to the mKP hierarchy.

\bigskip
\bigskip
\textbf{Acknowledgements:}
This work is supported by the National Natural Science Foundation of China under Grant Nos. 12171133, 12271136 and 12171132, and the Anhui Province Natural Science Foundation No. 2008085MA05.

\end{document}